\documentclass[pdfa,a4paper,USenglish,cleveref,autoref,,thm-restate]{lipics-v2021}
\newcommand{\eps}{\varepsilon}
\newcommand{\diam}{\text{diam}}
\newcommand{\per}{\text{per}}
\newcommand{\proj}{\text{proj}}
\newcommand{\wdth}{\text{width}}
\newcommand{\hght}{\text{height}}
\newcommand{\aratio}{\text{aratio}}
\newcommand{\MST}{\text{MST}}

\hideLIPIcs

\title{On Euclidean Steiner \texorpdfstring{$(1+\eps)$}{(1+epsilon)}-Spanners}
\titlerunning{On Euclidean Steiner \texorpdfstring{\boldmath $(1+\eps)$}{(1+epsilon)}-Spanners}

\author{Sujoy Bhore}{Universit\'{e} Libre de Bruxelles, Brussels, Belgium}{sujoy.bhore@gmail.com}{https://orcid.org/0000-0003-0104-1659}{Research on this paper was supported by the Fonds de la Recherche Scientifique-FNRS under Grant no MISU F 6001.}
\author{Csaba D. T\'{o}th}{California State University Northridge, Los Angeles, CA, USA \and Tufts University, Medford, MA, USA}{csaba.toth@csun.edu}{https://orcid.org/0000-0002-8769-3190}{Research on this paper was partially supported by the NSF award DMS-1800734.}

\authorrunning{S. Bhore and C.\,D. T\'{o}th}
\Copyright{Sujoy Bhore and Csaba D. T\'{o}th}



\relatedversiondetails{Previous Version}{https://arxiv.org/abs/2010.02908}

\EventEditors{Markus Bl\"{a}ser and Benjamin Monmege}
\EventNoEds{2}
\EventLongTitle{38th International Symposium on Theoretical Aspects of Computer Science (STACS 2021)}
\EventShortTitle{STACS 2021}
\EventAcronym{STACS}
\EventYear{2021}
\EventDate{March 16--19, 2021}
\EventLocation{Saarbr\"{u}cken, Germany (Virtual Conference)}
\EventLogo{}
\SeriesVolume{187}
\ArticleNo{13}
\nolinenumbers

\begin{document}

\maketitle
\vspace{0.5\baselineskip}
\enlargethispage{-0.5\baselineskip}
\begin{abstract}
Lightness and sparsity are two natural parameters for Euclidean $(1+\eps)$-spanners.
Classical results show that, when the dimension $d\in \mathbb{N}$ and $\eps>0$ are constant,
every set $S$ of $n$ points in $d$-space admits an $(1+\eps)$-spanners with $O(n)$ edges and weight proportional to that of the Euclidean MST of $S$. Tight bounds on the dependence on $\eps>0$ for constant $d\in \mathbb{N}$ have been established only recently.
Le and Solomon (FOCS 2019) showed that Steiner points can substantially improve the lightness and sparsity of a $(1+\eps)$-spanner. They gave upper bounds of $\tilde{O}(\eps^{-(d+1)/2})$ for the minimum lightness in dimensions $d\geq 3$, and $\tilde{O}(\eps^{-(d-1))/2})$ for the minimum sparsity in $d$-space for all $d\geq 1$. They  obtained lower bounds only in the plane ($d=2$). Le and Solomon (ESA 2020) also constructed Steiner $(1+\eps)$-spanners of lightness $O(\eps^{-1}\log\Delta)$ in the plane, where $\Delta\in  \Omega(\sqrt{n})$ is the \emph{spread} of $S$, defined as the ratio between the maximum and minimum distance between a pair of points.

In this work, we improve several bounds on the lightness and sparsity of Euclidean Steiner $(1+\eps)$-spanners. Using a new geometric analysis, we establish lower bounds of $\Omega(\eps^{-d/2})$ for the lightness and $\Omega(\eps^{-(d-1)/2})$ for the sparsity of such spanners in Euclidean $d$-space for all $d\geq 2$. We use the geometric insight from our lower bound analysis to construct
Steiner $(1+\eps)$-spanners of lightness $O(\eps^{-1}\log n)$ for $n$ points in Euclidean plane.
\end{abstract}

\section{Introduction}\label{sec:intro}

For an edge-weighted graph $G$, a subgraph $H$ of $G$ is a $t$-\emph{spanner} if $\delta_H(u,v) \le t\cdot \delta_G(u,v)$, where $\delta_G(u,v)$ denotes the shortest path distance between vertices $u$ and $v$.  A subgraph $H$ of $G$ is a \emph{$t$-spanner}, for some $t\ge 1$, if for every $pq\in \binom{V(G)}{2}$, we have $d_G(p,q)\leq t\cdot w(pq)$. The parameter $t$ is called the \emph{stretch factor} of the spanner.
Spanners are fundamental graph structures with many applications in the area of
distributed systems and communication, distributed queuing protocol, compact routing schemes, etc.; see~\cite{demmer1998arrow,
herlihy2001competitive,
PelegU89a,
peleg1989trade}.
Two important parameters of a spanner $H$ are \emph{lightness} and \emph{sparsity}.
The \emph{lightness} of $H$ is the ratio $w(H)/w(\MST)$ between the total weight of $H$ and the weight of a minimum spanning tree (MST). The \emph{sparsity} of $H$ is the ratio $|E(H)|/|E(\MST)|\approx |E(H)|/|V(G)|$ between the number of edges of $H$ and an MST. As $H$ is connected, the trivial lower bound for both the lightness and the sparsity of a spanner is 1.
When the vertices of $G$ are points in a metric space, the edge weights obey the triangle inequality. The most important examples include Euclidean $d$-space and, in general, metric spaces with constant doubling dimensions (the doubling dimension of $\mathbb{R}^d$ is $d$).

In a \emph{geometric spanner}, the underlying graph $G=(S,\binom{S}{2})$ is the complete graph on a finite point set $S$ in $\mathbb{R}^d$, and the edge weights are the Euclidean distances between vertices.
Euclidean spanners are one of the fundamental geometric structures that find application across domains, such as, topology control in wireless networks~\cite{schindelhauer2007geometric}, efficient regression in metric spaces~\cite{gottlieb2017efficient},
approximate distance oracles~\cite{gudmundsson2008approximate}, and others. Rao and Smith~\cite{rao1998approximating} showed the relevance of Euclidean spanners in the context of other geometric \textsf{NP}-hard problems, e.g., Euclidean traveling salesman problem and Euclidean minimum Steiner tree problem, and introduced the so called \emph{banyans}\footnote{A
$(1+\eps)$-banyan for a set of points $A$ is a set of points $A'$
and line segments $S$ with
endpoints in $A\cup A'$ such that a
$1+\eps$ optimal Steiner Minimum Tree for any subset of $A$ is
contained in $S$}, which is a generalization of graph spanners.
Apart from lightness and sparsity, various other optimization criteria have been considered, e.g.,
bounded-degree spanners~\cite{bose2005constructing} and $\alpha$-diamond spanners~\cite{das1989triangulations}. Several distinct construction approaches have been developed for Euclidean spanners, that each found further applications in geometric optimization, such as
well-separated pair decomposition (WSPD) based spanners~\cite{callahan1993optimal, GudmundssonLN02}, skip-list spanners~\cite{arya1994randomized},
path-greedy and gap-greedy spanners~\cite{althofer1993sparse, arya1997efficient}, and more. For an excellent survey of  results and techniques on Euclidean spanners up to 2007, we refer to the book by Narasimhan and Smid~\cite{narasimhan2007geometric}.

\subparagraph{Sparsity.}
A large body of research on spanners has been devoted to \emph{sparse spanners} where the objective is to obtain a spanner with small number edges, preferably $O(|S|)$, with $1+\eps$ stretch factor, for any given $\eps>0$.
Chew~\cite{Chew86} was the first to show that there exists a Euclidean spanner with a linear number of edges and \emph{stretch factor} $\sqrt{10}$. The stretch factor was later improved to $2$~\cite{Chew89}. Clarkson~\cite{Clarkson87} designed the first Euclidean $(1+\eps)$-spanner, for arbitrary small $\eps>0$; an alternative algorithm was  presented by Keil~\cite{keil1988approximating}.
Later, Keil and Gutwin~\cite{keil1992classes} showed that the Delaunay triangulation of the point set $S$ is a $2.42$-spanner. Moreover, these papers introduced the fixed-angle $\Theta$-graph\footnote{The $\Theta$-graph is a type of geometric spanner similar to Yao graph~\cite{yao1982constructing}, where the space around each point $p\in P$ is partitioned into cones of angle $\Theta$, and $S$ will be connected to a point $q\in P$ whose orthogonal projection to some fixed ray contained in the cone is closest to $S$.} as a potential new tool for designing spanners in $\mathbb{R}^2$, which was later generalized to higher dimension by Ruppert and Seidel~\cite{ruppert1991approximating}. One can construct an $(1+\eps)$-spanner with $O(n\eps^{-d+1})$ edges by taking the angle $\Theta$ to be proportional to $\eps$ in any constant dimension $d\geq 1$. A fundamental question in this area is whether the trade-off between the stretch factor $1+\eps$ and the sparsity $O(n\eps^{-d+1})$ is tight.

\subparagraph{Lightness.}
For a set of points $S$ in a metric space, the  lightness is the ratio of the spanner weight (i.e., the sum of all edge weights) to the weight of the minimum spanning tree $\MST(S)$.
Das et al.~\cite{das1993optimally} showed that \emph{greedy-spanner} (\cite{althofer1993sparse}) has constant lightness in $\mathbb{R}^3$. This was generalized later to $\mathbb{R}^d$, for all $d\in \mathbb{N}$, by Das et al.~\cite{narasimhan1995new}. However the dependencies on $\eps$ and $d$ has not been addressed. Rao and Smith showed that the greedy spanner has lightness $\eps^{-O(d)}$ in $\mathbb{R}^d$ for every constant $d$, and asked what is the best possible constant in the exponent. A complete proof for $(1+\eps)$-spanner with lightness $O(\eps^{-2d})$ is in the book on geometric spanners~\cite{narasimhan2007geometric}. Recently, Borradaile et al.~\cite{borradaile2019greedy} showed that the greedy $(1+\eps)$-spanner of a finite metric space of doubling dimension $d$ has lightness $\eps^{-O(d)}$.

\subparagraph{Dependence on \texorpdfstring{\boldmath $\eps>0$}{epsilon > 0} for constant dimension \texorpdfstring{\boldmath $d$}{d}.}
The dependence of the lightness and sparsity on $\eps>0$ for constant $d\in \mathbb{N}$ has been studied only recently. Le and Solomon~\cite{le2019truly} constructed, for every $\eps>0$ and constant $d\in \mathbb{N}$, a set $S$ of $n$ points in $\mathbb{R}^d$ for which any $(1+\eps)$-spanner must have lightness $\Omega(\eps^{-d})$ and sparsity $\Omega(\eps^{-d+1})$, whenever $\eps = \Omega(n^{-1/(d-1)})$. Moreover, they showed that the greedy $(1+\eps)$-spanner in $\mathbb{R}^d$ has lightness $O(\eps^{-d}\log \eps^{-1})$.

Steiner points are additional vertices in a network (via points) that are not part of the input, and a $t$-spanner must achieve stretch factor $t$ only between pairs of the input points in $S$. A classical problem on Steiner points arises in the context of minimum spanning trees. The \emph{Steiner ratio} is the supremum ratio between the weight of a \emph{minimum Steiner tree} and a \emph{minimum spanning tree} of a finite point set, and it is at least $\frac{1}{2}$ in any metric space due to triangle inequality.

Le and Solomon~\cite{le2019truly} noticed that Steiner points can substantially improve the bound on the lightness and sparsity of an $(1+\eps)$-spanner. Previously, Elkin and Solomon~\cite{elkin2015steiner} and Solomon~\cite{Solomon15} showed that Steiner points can improve the weight of the network in the single-source setting. In particular, the so-called \emph{shallow-light trees} (\textsf{SLT}) is a single-source spanning tree that concurrently approximates a shortest-path tree (between the source and all other points) and a minimum spanning tree (for the total weight). They proved that Steiner points help to obtain exponential improvement on the lightness \textsf{SLT}s in a general metric space~\cite{elkin2015steiner}, and quadratic improvement on the lightness in Euclidean spaces~\cite{Solomon15}.

Le and Solomon, used Steiner points to improve the bounds for lightness and sparsity of Euclidean spanners. For minimum sparsity, they gave an upper bound of $O(\eps^{(1-d)/2})$ for $d$-space and a lower bound of $\Omega(\eps^{-1/2}/\log\eps^{-1})$ in the plane ($d=2$)~\cite{le2019truly}.
For minimum lightness, Le and Solomon~\cite{le2020light} gave an upper bound of $O(\eps^{-1}\log\Delta)$ in the plane and $O(\eps^{-(d+1)/2}\log\Delta)$ in dimension $d\geq 3$, where $\Delta$ is the \emph{spread} of the point set, defined as the ratio between the maximum and minimum distance between a pair of points. Note that in any space with doubling dimension $d$ (including $\mathbb{R}^d$), we have $\log \Delta\geq \Omega(\log_d n)$, but the spread $\Delta$ is in fact unbounded.
Very recently, Le and Solomon~\cite{le2020unified} constructed Steiner $(1+\eps)$-spanners with lightness $\tilde{O}(\eps^{-(d+1)/2})$ in dimensions $d\geq 3$.

\subparagraph{Our Contributions.}
In this work, we improve the bounds on the lightness and sparsity of Euclidean Steiner $(1+\epsilon)$-spanners. First, in Section~\ref{sec:lower}, we prove the following lower bounds.

\begin{restatable}{theorem}{lowerboundth}
\label{thm:lb}
Let a positive integers $d$ and real $\varepsilon>0$ be given such that $\eps \leq 1/d$.
Then there exists a set $S$ of $n$ points in $\mathbb{R}^d$
such that any Euclidean Steiner $(1+\eps)$-spanner for $S$ has
lightness $\Omega(\eps^{-d/2})$ and sparsity $\Omega(\eps^{(1-d)/2})$.	
\end{restatable}

For lightness in dimension $d=2$, this improves the earlier bound of $\Omega(\eps^{-1}\log^{-1}(\eps^{-1}))$ by Le and Solomon~\cite{le2019truly} by a logarithmic factor; and it is the first lower bound in dimensions $d\geq 3$. The point set $S$ in Theorem~\ref{thm:lb} is fairly simple, it consists of two square grids in two parallel hyperplanes in $\mathbb{R}^d$. However, our lower-bound analysis is significantly simpler than that of~\cite{le2019truly}. In particular, our analysis does not depend on planarity, and it generalizes to higher dimensions. The key new insight pertains to a geometric property of Steiner $(1+\eps)$-spanners: If the length of an $ab$-path $S$ between points $a,b\in\mathbb{R}^d$ is at most $(1+\eps)\|ab\|$, then ``most'' of the edges of $S$ are almost parallel to $ab$. We expand on this idea in Section~\ref{sec:pre}.

Then, in Section~\ref{sec:upper} we prove the following theorem on light spanners.

\begin{restatable}{theorem}{upperboundtheorem}
\label{thm:upper}
For every set $S$ of $n$ points in Euclidean plane,
there exists a Steiner $(1+\eps)$-spanner of
lightness $O(\eps^{-1}\log n)$.
\end{restatable}

\enlargethispage{1.5\baselineskip}
This result improves on an earlier bound of $O(\eps^{-1}\log \Delta)$ by Le and Solomon~\cite{le2020light}, where $\Delta$ is the \emph{spread} of the point set, defined as the ratio between the maximum and minimum distance between a pair of points. Note that $\Delta\geq \Omega(\log n)$ in every metric space of constant doubling dimension.
Recently, Le and Solomon~\cite{le2020light2} noted in the revised version of their paper that the $\log \Delta$ factor can be reduced to a $\log n$ factor by a general discretization technique (see, e.g., Chan et al.~\cite{DBLP:journals/siamcomp/ChanLNS15}). Very recently, Bhore and T\'{o}th~\cite{DBLP:journals/corr/abs-2012-02216} achieved the optimal dependence on $\eps$ and showed that, for every finite points set $S\subset \mathbb{R}^2$ and $\eps>0$, there exists a Euclidean Steiner $(1+\eps)$-spanner of weight $O(\frac{1}{\eps}\,\|\MST(S)\|)$. The spanner construction in~\cite{DBLP:journals/corr/abs-2012-02216} is a far-reaching generalization of the methods we develop in the proof of Theorem~\ref{thm:upper}. In particular, both papers use \emph{directional spanners}, introduced in~Section~\ref{sec:upper} of this paper, as a key ingredient and construct a Euclidean Steiner $(1+\eps)$-spanner as a union of $O(\eps^{-1/2})$ directional spanners.

\section{Preliminaries}\label{sec:pre}

Let $d\geq 2$ be an integer, and $S$ a set of $n$ points in $\mathbb{R}^d$. For $a,b\in \mathbb{R}^d$, the Euclidean distance between $a$ and $b$ is denoted by $\|ab\|$. For a set $E$ of line segments in $\mathbb{R}^d$, let $\|E\|=\sum_{e\in E}\|e\|$ be the total weight of all segments in $E$. For a geometric graph $G=(S,E)$, where $S\subset \mathbb{R}^d$, we also use the notation $\|G\|=\|E\|$, which is the Euclidean weight of graph $G$.

We briefly review a few geometric primitives in $d$-space.
For $a,b\in \mathbb{R}^d$, the locus of points $c\in \mathbb{R}^d$ with $\|ac\|+\|cb\|\leq (1+\eps)\|ab\|$ is an ellipsoid $\mathcal{E}_{ab}$ with foci $a$ and $b$, and major axis of length $(1+\eps)\|ab\|$; see Fig.~\ref{fig:ellipse}(a).
Note that all $d-1$ minor axes of $\mathcal{E}_{ab}$ are $\sqrt{(1+\eps)^2-1}\|ab\|=\sqrt{2\eps+\eps^2}\|ab\|<\sqrt{3\eps}\|ab\|$ when $\eps<1$.
In particular, the aspect ratio of the minimum bounding box of $\mathcal{E}_{ab}$ is roughly  $\sqrt{\eps}$. By the triangle inequality, $\mathcal{E}_{ab}$ contains every  $ab$-path of weight at most $(1+\eps)\|ab\|$.

The unit vectors in $\mathbb{R}^d$ are on the $(d-1)$-sphere $\mathbb{S}^{d-1}$; the direction vectors of a line in $\mathbb{R}^d$ can be represented by vectors of a hemisphere. The \emph{angle} between two unit vectors, $\overrightarrow{u}_1$ and $\overrightarrow{u}_2$ is
$\angle(\overrightarrow{u}_1,\overrightarrow{u}_2)=\arccos(\overrightarrow{u}_1\cdot \overrightarrow{u}_2)\in (-\pi,\pi)$. Between two (undirected) edges $e_1$ and $e_2$ with unit direction vectors $\pm\overrightarrow{u}_1$ and $\pm\overrightarrow{u}_2$, we define the angle as $\angle (e_1,e_2)=\arccos | \overrightarrow{u}_1\cdot \overrightarrow{u}_2|\in [0,\pi)$. Let $\proj_{ab}(e)$ denote the orthogonal projection of an edge $e$ to the supporting line of $ab$, see Fig.~\ref{fig:ellipse}(b); and note that $\|\proj_{ab}(e)\|=\|e\|\cos(ab,e)$.

\begin{figure}[htbp]
\centering
\includegraphics[width=\textwidth]{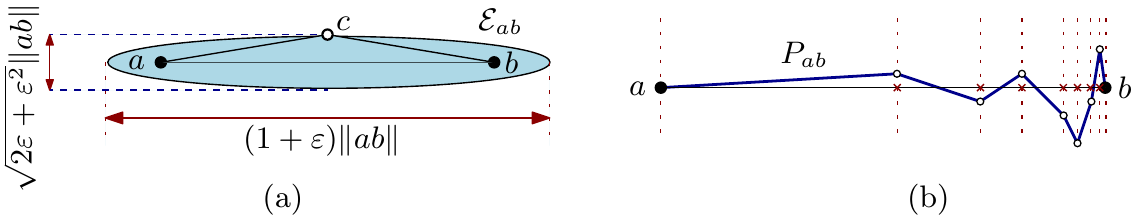}
\caption{(a) An ellipse $\mathcal{E}_{ab}$ with foci $a$ and $b$, and major axis $(1+\eps)\|ab\|$.
(b) A monotone $ab$-path $P_{ab}$, and the projections of its edges to $ab$.}
\label{fig:ellipse}
\end{figure}

\subparagraph{Characterization for Short \texorpdfstring{\boldmath $ab$}{ab}-Paths.}
Let $a,b\in \mathbb{R}^d$, and let $P_{ab}$ be a polygonal $ab$-path of weight at most $(1+\eps)\|ab\|$. We show that ``most'' edges along $P_{ab}$ must be ``nearly'' parallel to $ab$. Specifically, for an angle $\alpha\in [0,\pi/2)$, we distinguish between two types of edges in $P_{ab}$. Denote by $E(\alpha)$ the set of edges $e$ in $P_{ab}$ with $\angle(ab,e)\leq\alpha$; and let $F(\alpha)$ be the set of all other edges of $P_{ab}$.
Clearly, we have $\|P_{ab}\|=\|E(\alpha)\|+\|F(\alpha)\|$ for all $\alpha$.

\begin{lemma}\label{lem:parallel}
Let $a,b\in \mathbb{R}^d$ and let $P_{ab}$ be an $ab$-path of weight $\|P_{ab}\|\leq (1+\eps)\|ab\|$. Then for every $i\in \{1,\ldots, \lfloor1/\sqrt{\eps}\rfloor\}$,
we have $\|E(i\cdot \sqrt{\eps})\|\geq (1-2/i^2)\,\|ab\|$.
\end{lemma}
\begin{proof}
Suppose, to the contrary, that $\|E(i\cdot \sqrt{\eps})\|<(1-2/i^2)\,\|ab\|$ for an $i\in \{1,\ldots, \lfloor1/\sqrt{\eps}\rfloor\}$. We have
\begin{equation}\label{eq:proj}
\sum_{e\in E(i\,\sqrt{\eps})\cup F(i\,\sqrt{\eps})}  \| \text{proj}_{ab}(e)\|\geq \|ab\|,
\end{equation}
which implies
\begin{align}\label{eq:proj2}
\sum_{e\in F(i\,\sqrt{\eps})} \| \text{proj}_{ab}(e)\|
&\geq  \|ab\|-\sum_{e\in E(i\,\sqrt{\eps})} \| \text{proj}_{ab}(e)\| \\
&\geq  \|ab\|-\sum_{e\in E(i\,\sqrt{\eps})} \| e\| \nonumber\\
&= \|ab\|-\|E(i\,\sqrt{\eps})\|. \nonumber
\end{align}
Recall that for every edge $e\in F(i\,\sqrt{\eps})$, we have $\angle(e,ab)\geq i\cdot \sqrt{\eps}$. Using the Taylor estimate $\frac{1}{\cos(x)}\geq 1+\frac{x^2}{2}$, for every $e\in F(i\,\sqrt{\eps})$, we obtain
\[|e\|\geq \frac{\|\text{proj}_{ab}(e)\|}{\cos(i\cdot \sqrt{\eps})}
\geq \|\text{proj}_{ab}(e)\|\left(1+\frac{(i\,\sqrt{\eps})^2}{2}\right)
= \|\text{proj}_{ab}(e)\|\left(1+\frac{i^2\,\eps}{2}\right),\]
Combined with \eqref{eq:proj2}, this yields
\begin{align*}
\|P_{ab}\|
&= \sum_{e\in E(i\,\sqrt{\eps})} \|e\| +  \sum_{e\in F(i\,\sqrt{\eps})} \|e\| \\
&\geq \sum_{e\in E(i\,\sqrt{\eps})} \|e\| +  \sum_{e\in F(i\,\sqrt{\eps})} \|\text{proj}_{ab}(e)\|\left(1+\frac{i^2\,\eps}{2}\right) \\
&\geq \|E(i\,\sqrt{\eps})\| +\big(\|ab\|-\|E(i\,\sqrt{\eps})\|\big) \left(1+\frac{i^2\,\eps}{2}\right)\\
&= \left(1+\frac{i^2\,\eps}{2}\right) \|ab\|
- \frac{i^2\,\eps}{2}\, \|E(i\,\sqrt{\eps})\|\\
&> \left(1+\frac{i^2\,\eps}{2}\right)  \|ab\|
-\frac{i^2\,\eps}{2} \left(1-\frac{2}{i^2}\right)\,\|ab\|\\
&\geq \left(1+\frac{i^2\,\eps}{2}\right) \|ab\| -\left(\frac{i^2}{2}-1\right)\eps\,\|ab\| \\
&= (1+\eps) \|ab\|,
\end{align*}
which is a contradiction.
\end{proof}

We use Lemma~\ref{lem:parallel} in the analysis of our lower bound construction in Section~\ref{thm:lb}.
We can also derive a converse of Lemma~\ref{lem:parallel} for monotone $ab$-paths. An  $ab$-path is \emph{monotone} if $\angle(\overrightarrow{ab},\overrightarrow{e})>0$ for every directed edge $\overrightarrow{e}$ of $P_{ab}$, where the path is directed from $a$ to $b$. Equivalently, an $ab$-path is monotone if it crosses every hyperplane orthogonal to $ab$ at most once.
We show that if the angle $\angle(\overrightarrow{ab},\overrightarrow{d})$ is sufficiently small for ``most'' of the edges of $P_{ab}$, then $\|P_{ab}\|\leq (1+\eps)\|ab\|$.

\begin{lemma}\label{lem:parallel+}
For every $\delta>0$, there is a $\kappa>0$ with the following property.
For $a,b\in \mathbb{R}^d$ and a monotone an $ab$-path $P_{ab}$,
if $\|F(i\,\sqrt{\eps\kappa})\|\leq \|P_{ab}\|/i^{2+\delta}$ for all $i\in\{1,\ldots ,\lceil \pi/\sqrt{\eps\kappa}\rceil\}$, then $\|P_{ab}\|\leq (1+\eps)\|ab\|$.
\end{lemma}
\begin{proof}
Let $P_{ab}$ be an $ab$-path with edge set $E$. Note that, by definition, $F(0)=E$.
For angles $0\leq \alpha<\beta\leq \pi/2$, let $E(\alpha,\beta)$ denote the set of edges $e\in E$ with $\alpha\leq \angle (ab,e)<\beta$. For convenience, we put $m=\lceil \pi/\sqrt{\eps\kappa}\rceil$. Using the Taylor estimate $\cos x\geq 1-x^2/2$, we can bound the excess weight of $P_{ab}$ as follows.
\begin{align*}
\|P_{ab}\|-\|ab\|
&= \sum_{e\in E} \|e\| - \sum_{e\in E} \|\proj_{ab}e\|\\
&= \sum_{e\in E} \|e\| (1-\cos \angle (ab,e))\\
&\leq \sum_{i=1}^m \|E((i-1)\sqrt{\eps\kappa},i\,\sqrt{\eps\kappa})\| (1-\cos (i\,\sqrt{\eps\kappa})) \\
&\leq \sum_{i=1}^m \|E((i-1)\sqrt{\eps\kappa},i\,\sqrt{\eps\kappa})\| \cdot \frac{i^2\,\eps\kappa}{2} \\
&\leq \sum_{i=1}^m \left(\|F((i-1)\sqrt{\eps\kappa})\| - \|F(i\,\sqrt{\eps\kappa})\|\right)
\cdot \frac{i^2\,\eps\kappa}{2} \\
&= F(0)\cdot \frac{1^2\eps\kappa}{2}+\sum_{i=1}^m \|F(i\,\sqrt{\eps\kappa})\|
\left(\frac{(i+1)^2\,\eps\kappa}{2}-  \frac{i^2\,\eps\kappa}{2}\right) \\
&\leq \|P_{ab}\|\cdot \frac{\eps\kappa}{2}     +\sum_{i=1}^m \frac{\|P_{ab}\|}{i^{2+\delta}}\cdot \frac{(2i+1)\eps\kappa}{2} \\
&\leq \frac{\eps\kappa}{2}\cdot \|P_{ab}\| \left(1+\sum_{i=1}^\infty \frac{2i+1}{i^{2+\delta}}\right)
\end{align*}
For $\kappa= 2(1+\sum_{i=1}^\infty (2i+1)/2^{2+\delta})^{-1}$, we obtain
\[\|P_{ab}\|-\|ab\|\leq \frac{\eps}{2}\,\|P_{ab}\|,\]
which readily implies $\|P_{ab}\|\leq (1-\eps/2)^{-1}\|ab\|<(1+\eps)\|ab\|$,
as required.
\end{proof}

The criteria in Lemma~\ref{lem:parallel+} can certify that a geometric graph $G$ is a Euclidean Steiner $(1+\eps)$-spanner for a point set $S$. Intuitively, an $(1+\eps)$-spanner should contain, for all point pairs $a,b\in S$, an $ab$-path in which the majority of edges $e$ satisfy $\angle(ab,e)\leq O(\sqrt{\eps})$, with exceptions quantified by Lemma~\ref{lem:parallel+}. This property has already been used by Solomon~\cite{Solomon15} in the single-source setting, for the design of shallow-light trees.
We use shallow-light trees in our upper bound (Section~\ref{sec:upper}), instead of Lemma~\ref{lem:parallel+}. However, the characterization of $ab$-paths of weight at most $(1+\eps)\|ab\|$, presented in this section, may be of independent interest.

\begin{figure}[htbp]
\centering
\includegraphics[width=0.3\textwidth]{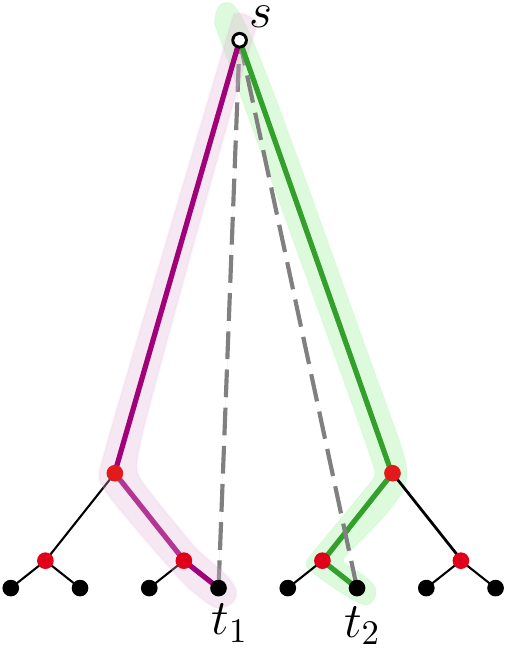}
\caption{An illustration of a shallow-light tree for a source $s$ and a set $S$ of collinear points. The input points and the Steiner points are colored black and red, respectively.}
\label{fig:slt}
\end{figure}

\subparagraph{Shallow-light trees.}
Shallow-light trees (\textsf{SLT}) were introduced by Solomon~\cite{Solomon15}.
Given a source $s$ and a point set $S$ in $\mathbb{R}^d$, an \emph{$(\alpha,\beta)$-SLT} is a Steiner tree rooted at $s$ that contains a path of weight at most $\alpha \,\|ab\|$ between the \emph{source} $s$ and any point $t\in S$, and has weight $\beta\,\|\MST(S)\|$. For our upper bounds in Section~\ref{sec:upper}, we use the following variant of shallow-light trees,
between a source $s$ and a set $S$ of collinear points in the plane; see Figure~\ref{fig:slt}.

\begin{lemma}[Solomon~{\cite[Section~2.1]{Solomon15}}]\label{lem:shallow}
Let $0<\eps<1$, let $S$ be a set of points in the $[-\frac12,\frac12]$ interval of the $x$-axis, and let $s=(0,\eps^{-1/2})$ be a point on the $y$-axis. Then there exists a geometric graph of weight $O(\eps^{-1/2})$ that contains, for every point $t\in S$, an $st$-path $P_{st}$ with $\|P_{st}\|\leq (1+\eps)\,\|st\|$.
\end{lemma}

\section{Lower Bounds}\label{sec:lower}

In this section we prove the following lower bound on the lightness of Euclidean Steiner spanners in $\mathbb{R}^d$. Our lower bound construction is a direct generalization of the 2-dimensional lower bound construction by Le and Solomon~\cite{le2019truly}. However, our analysis is significantly simpler than that of~\cite{le2019truly} and it does not depend on planarity. As a result, it easily extends to higher dimensions.

\lowerboundth*

\begin{proof}
First we establish the result for a point set of size $\Theta_d(\eps^{(1-d)/2})$ and then generalize to arbitrary $n$.
Let $Q=[0,1]^d$ be a unit cube in $\mathbb{R}^d$; see Fig.~\ref{fig:twogrids}.
The point set $S$ will consist of two square grids in two opposite faces of $Q$, with $d/\sqrt{\eps}$ spacing. Specifically, consider the lattice $L=(d/\sqrt{\eps})\cdot \mathbb{Z}^d$. Let $Q_0$ and $Q_1$, respectively, be the two faces of $Q$ orthogonal to the $x_d$-axis. Now let $S_0=L\cap Q_0$ and $S_1=L\cap Q_1$. We have $|S_0|=|S_1|=\lfloor 1/(d\sqrt{\eps})\rfloor^{d-1}=\Theta_d(\eps^{(1-d)/2})$, hence $|S|=\Theta_d(\eps^{(1-d)/2})$.

\begin{figure}[htbp]
\centering
\includegraphics[width=0.44\textwidth]{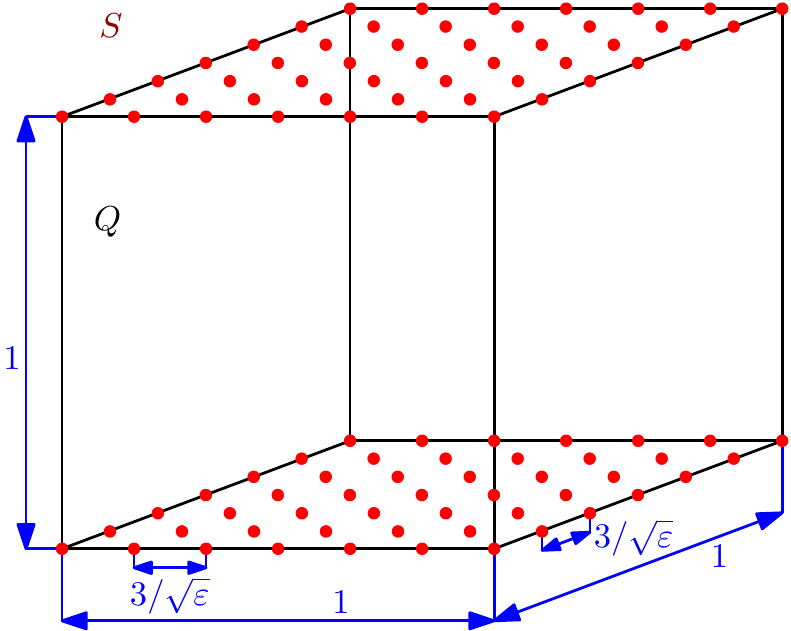}
\caption{A schematic image of $S$ in $\mathbb{R}^3$.}
\label{fig:twogrids}
\end{figure}

\enlargethispage{\baselineskip}
Let $N$ be a Euclidean Steiner $(1+\eps)$-spanner for $S$. For a point pair $(a,b)\in S_0\times S_1$, we have $1\leq \|ab\|\leq \text{diam}(Q)=\sqrt{d}$. The spanner $N$ contains an $ab$-path $P_{ab}$ of weight at most $(1+\eps)\|ab\|$, which lies in the ellipsoid $\mathcal{E}_{ab}$ with foci at $a$ and $b$, and major axis $(1+\eps)\|ab\|$. The ellipsoid $\mathcal{E}_{ab}$ is, in turn, contained in an infinite cylinder $C_{ab}$ with axis $ab$ and radius $\frac{\sqrt{(1+\eps)^2-1^2}}{2}\|ab\|<\sqrt{\eps}\|ab\| \leq  \sqrt{\eps d}$. The intersection of the cylinder $C_{ab}$ with hyperplanes containing $Q_0$ and $Q_1$, resp., is an ellipsoid of half-diameter at most $\sqrt{d-1}\cdot \sqrt{\eps d}<\sqrt{\eps}\,d$, and their centers are  $a$ and $b$, respectively. In particular, all point in $S$, other than $a$ and $b$, are in the exterior of $C_{ab}$.

We distinguish between two types of edges in the $ab$-path $P_{ab}$.
An edge $e$ of $P_{ab}$ is \emph{near-parallel} to $ab$ if $\angle(ab,e)<2\cdot\sqrt{\eps}$.
Let $E(ab)$ be the set of edges of $P_{ab}$ that are near-parallel to $ab$, and $F(ab)$ the set of all other edges of $P_{ab}$. Lemma~\ref{lem:parallel} with $i=2$ yields
\begin{equation}\label{eq:12}
\|E(ab)\|\geq \frac12 \|ab\|\geq \frac12.
\end{equation}

Notice that for two pairs $(a_1,b_1), (a_2,b_2)\in S_0\times S_1$, if $\{a_1,b_1\}\neq \{a_2,b_2\}$, then $E(a_1b_1)\cap E(a_2b_2)=\emptyset$. If $\angle(a_1 b_1,a_2 b_2) \geq 4\sqrt{\eps}$, this follows from the fact that the \emph{directions} near-parallel to $a_1b_1$ and $a_2b_2$, resp., are disjoint. If $\angle(a_ 1b_1,a_2 b_2) < 4\sqrt{\eps}$, then $a_1 b_1$ and $a_2 b_2$ are parallel, consequently the cylinders $C_{a_1 b_1}$ and $C_{a_2 b_2}$ have disjoint interiors, and so $E(a_1b_1)\cap E(a_2b_2)=\emptyset$.
Combined with \eqref{eq:12}, this yields
\begin{equation}\label{eq:weight}
\|N\|
\geq \sum_{(a,b)\in S_0\times S_1} \|E(ab)\|
\geq |S_0|\cdot |S_1|\cdot \frac12
\geq \Theta_d( \eps^{1-d}).
\end{equation}

Similarly to~\cite[Claim~5.3]{le2019truly}, we may assume that $N\subseteq Q$ (indeed, we can replace every vertex of $N$ outside of $Q$ by the closest point in the boundary of $Q$; such replacements do not increase the weight of $N$).
In follows that the weight of every edge is at most $\text{diam}(Q)=\sqrt{d}$.
Consequently,
\[|E(N)|\geq \frac{\|N\|}{\max_{e\in E(N)} \|e\|} =\frac{\Omega_d(\eps^{1-d})}{\sqrt{d}}=\Omega_d(\eps^{1-d}).\]
The sparsity of $N$ is
$|E(N)|/|S|=\Omega_d(\eps^{1-d}/\eps^{(1-d)/2})=\Omega_d(\eps^{(1-d)/2})$, as required.

The MST for the point set $S$ contains one unit-weight edge between $S_0$ and $S_1$, and the remaining $|S|-2$ edges each have weight $d\sqrt{\eps}$, which is the minimum distance between lattice points in $L$ (see~\cite{SteeleS89} for the asymptotic behavior of the MST of a section of the lattice). Therefore $\|\MST(S)\|=1+(|S|-2)d\sqrt{\eps}=\Theta_d(\eps^{1-d/2})$.
It follows that the lightness of $N$ is $\|N\|/\|\MST(S)\|=\Omega_d(\eps^{1-d}/\eps^{1-d/2})=\Omega_d(\eps^{-d/2})$, as claimed.
This completes the proof when $n=\Theta_d(\eps^{(1-d)/2})$.
\medskip\\
\textbf{\sffamily General Case.}
Finally, if $n\geq |S|=\Theta_d(\eps^{(1-d)/2})$, we can generalize
the above construction by duplication. Assume that $n=k\, |S|$ for some integer $k\geq 1$.
Let $Q_1,\ldots ,Q_k$, be disjoint axis-aligned unit hypercubes,
such that they each have an edge along the $x_1$-axis,
and two consecutive cubes are at distance 3 apart.
Let $S'$ be the union of $k$ translates of the point set $S$,
on the boundaries of $Q_1,\ldots, Q_k$.
Let $N'$ be a Euclidean Steiner $(1+\eps)$-spanner for $S'$.

\enlargethispage{-\baselineskip}
Since the ellipses induced by point pairs in different copies of $S$ are disjoint,
we have $\|N'\|\geq k\, \|N\|=\Omega_d(k\eps^{1-d})$ and $|E(N')|\geq k\, |E(N)|$.
This immediately implies that the sparsity of $N'$ is $|E(N')|/n  =|E(N)|/|S| \geq \Omega_d(\eps^{(1-d)/2})$.

The MST of $S'$ consists of $k$ translated copies of $\text{MST}(S)$ and
$k-1$ edges of weight 3 between consecutive copies. That is,
$\|\MST(S')\|=k\, \|\text{MST}(S)\|+3(k-1)=\Theta_d(k\eps^{1-d/2})$.
It follows that the lightness of $N'$ is $\Omega_d(\eps^{-d/2})$, as claimed.
\end{proof}

\section{Upper Bound}
\label{sec:upper}

In this section, we prove Theorem~\ref{thm:upper} and construct, for every $\eps>0$ and every set of $n$ points in $\mathbb{R}^2$, a Euclidean Steiner $(1+\eps)$-spanner of lightness $O(\eps^{-1}\log n)$. This matches the lower bound of Theorem~\ref{thm:lb} up to a $O(\log n)$ factor, and improves upon the previous bound of $O(\eps^{-1}\log \Delta)$ by Le and Solomon~\cite[Theorem~1.2]{le2020light}.

\subparagraph{Directional \texorpdfstring{\boldmath $(1+\eps)$}{(1+epsilon)}-spanners.}
Let $S$ be a set of $n$ points in the plane. Assume, without loss of generality, that $\diam(S)\geq 1/2$ and $S\subset [0,1]^2$. Then the weight of the Euclidean spanning tree of $S$ is bounded by $\|\MST(S)\|\geq \diam(S)\geq 1/2$, and $\|\MST(S)\|\leq O(n^{1/2})$ by a classical result by Few~\cite{Few55}. Both bounds are tight in the worst case.
Note that the weight of the $\MST$ is in a rather broad range, which makes it challenging to bound the weight of the Steiner $(1+\eps)$-spanner by $O(\|\MST(S)\|\cdot \eps^{-1})$.

The \emph{direction} of a line $L$ in the plane is given by the counterclockwise angle $\theta\in [0,\pi)$ between the positive $x$-axis and $L$. A line segment $pq$ inherits its direction from its supporting line. The \emph{distance} between directions $\theta_1,\theta_2\in [0,\pi)$, of lines $L_1$ and $L_2$, resp., is \[\angle(L_1,L_2)=\min\{|\theta_1-\theta_2|, \pi-|\theta_1-\theta_2|\}.\]
For an interval $D\subset [0,\pi)$ of directions, we construct a  Euclidean Steiner $(1+\eps)$-spanner restricted to points pairs whose directions are in $D$.
We define a spanner in this restricted sense as follows.

\begin{definition}
Let $S$ be a finite point set in $\mathbb{R}^2$, and let $D\subset [0,\pi)$ be an set of directions. A geometric graph $G$ is a \emph{directional $(1+\eps)$-spanner} for $S$ and $D$ if
for every $a,b\in S$, where the direction of $ab$ is in $D$,
graph $G$ contains an $ab$-path of weight at most $(1+\eps)\|ab\|$.
\end{definition}

Our main lemma is the following.
\begin{restatable}{lemma}{directionallemma}
\label{lem:dir}
For a set $S$ of $n$ points in the plane, and for the interval  $D=[\frac{\pi-\sqrt{\eps}}{8},\frac{\pi+\sqrt{\eps}}{8}]$ of directions, there exists
a directional $(1+\eps)$-spanner $G$ of weight $O(\|\MST(S)\|\eps^{-1/2}\log n)$.
\end{restatable}

We prove Lemma~\ref{lem:dir} in Section~\ref{ssec:rectangulations} below.
Here we show that Lemma~\ref{lem:dir} implies Theorem~\ref{thm:upper}.

\upperboundtheorem*
\begin{proof}[Proof of Theorem~\ref{thm:upper}]
Let $S$ be a set of $n$ points in the plane. Let $\eps>0$ be given, and let $k=\Theta(\eps^{-1/2})$ be an integer.
We partition the space of directions into $k$ intervals as
\[[0,\pi)=\bigcup_{i=1}^{k} D_i,
\mbox{ \rm where }
D_i=\left[\frac{(i-1)\,\pi}{k},\frac{i\,\pi}{k}\right)
\mbox{ \rm for }
i\in \{1,\ldots , k\}.\]
By Lemma~\ref{lem:dir}, for $i=1,\ldots ,k$, there exists a geometric graph $N_i$ of weight $O(\|\MST\| \log (n)\cdot \eps^{-1/2})$ such that for every point pair $a,b\in S$, if the direction of $ab$ is in $D_i$, then $N_i$ contains an $ab$-path of weight at most $(1+\eps)\|ab\|$.

\begin{figure}[htbp]
\centering
\includegraphics[width=0.9\textwidth]{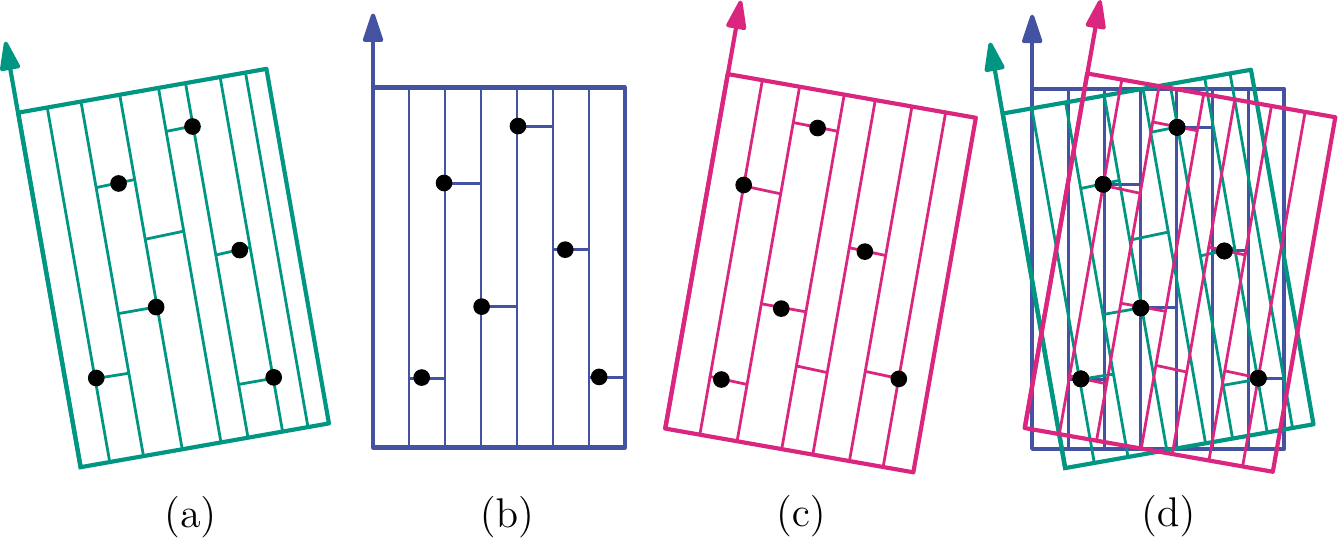}
\caption{(a)--(c) Schematic figures for directional spanners $N_1$, $N_2$, and $N_3$ for three disjoint intervals of directions, for a set of 6 points in the plane.
(d) The union $N_1\cup N_2\cup N_3$ of the networks.}
\label{fig:directional-rectangulation}
\end{figure}

Let $N=\bigcup_{i=1}^k N_i$ be the union of the networks $N_i$ for $i\in \{1,\ldots ,k\}$; see Figure~\ref{fig:directional-rectangulation} for an illustration.
Since $[0,\pi)=\bigcup_{i=1}^k D_i$, the graph $N$ contains a path of weight at most $(1+\eps)\|ab\|$ for all point pairs $a,b\in S$, and so $N$ is a Euclidean Steiner $(1+\eps)$-spanner for $S$.
The weight of $N$ is
\[\|N\|=\sum_{i=1}^k \|N_i\|
\leq k\cdot O\left(\frac{\|\MST(S)\|\log n}{\sqrt{\eps}}\right)
\leq O\left(\frac{\|\MST(S)\|\log n}{\eps}\right),\]
as required.
\end{proof}

The second component of our approach is a subdivision of the bounding box of $S$ into regions that are long and skinny in one particular direction. We start with discussing the special cases of a single rectangle.

\subsection{Rectangles}
\label{ssec:rectangle}

We illustrate our general strategy with a simple special case,
where the points in $S$ lie on the boundary of an axis-aligned rectangle.
We first assume that $R$ is narrow and tall, and we construct an
directional $(1+\eps)$-spanner for an interval of near-vertical directions.
For further reference, we define the \emph{aspect ratio} of an
axis-parallel rectangle $R$ by
\[\aratio(R)=\frac{\wdth(R)}{\hght(R)}.\]

\begin{lemma}\label{lem:rect1}
Let $R$ be an axis-aligned rectangle with $\frac18\, \sqrt{\eps}\leq \aratio(R)\leq \frac14\,\sqrt{\eps}$.
Then for a finite point set $S$ on the boundary of $R$, and for the interval  $D=[\frac{\pi-\sqrt{\eps}}{8},\frac{\pi+\sqrt{\eps}}{8}]$ of directions, there exists
a directional $(1+\eps)$-spanner $G$ with $\|G\|\leq O(\hght(R))$.
\end{lemma}
\begin{proof}
We construct a graph $G$ as a union of the boundary of $R$ and a finite number of
shallow-light trees. If both $a$ and $b$ are in the same side of $R$, then
the boundary of $R$ contains an $ab$-path of weight $\|ab\|$.
We next consider cases in which $a$ and $b$ are in different sides of $R$.
Since $\frac18\, \sqrt{\eps}\leq \aratio(R)$, if $a$ and $b$ are in the interior
of the left and right side of $R$, resp., then the direction of the segment $ab$
falls outside of $D$.

Let $c$ be the center of the left edge of $R$; refer to Fig.~\ref{fig:rectangle}(a).
Take two shallow-light trees between $c$ and each horizontal side of $R$ with
parameter $\eps'=\eps/4$. By Lemma~\ref{lem:shallow}, the weight of the two
trees is $O(\hght(R))$. If $a$ and $b$ are on the top and bottom sides of $R$, resp.,
then $\|ab\|\geq \hght(R)$.
The union of the two shallow-light trees rooted at $c$ contains
an $ab$-path of length
\begin{align*}
(1+\eps')(\|as\|+\|sb\|)
& \leq (1+\eps')\cdot 2\cdot \sqrt{\left(\frac{\hght(R)}{2}\right)^2+\left(\frac{\wdth(R)}{2}\right)^2}\\
&< \left(1+\frac{\eps}{4}\right)\left(1+\frac{\eps}{2}\right)\hght(R)\\
&< (1+\eps)\,\hght(R)\\
&\leq (1+\eps)\|ab\|.
\end{align*}

It remains to construct $ab$-paths for point pairs on adjacent sides of $R$.
We describe the construction for a left and bottom sides of $R$;
the constructions for all other pairs of adjacent sides is analogous,
and we can take the union of all constructions.
Without loss of generality, assume that the lower-left corner of $R$ is
the origin $o$; see Fig.~\ref{fig:rectangle}(b).
For every positive integer $i\in \mathbb{N}^+$, let $p_i=(0,\hght(R)/2^i)$,
and $q_i=(\wdth(R)/2^{i-1})$. Note that $p_i$ is on the left side of $R$,
and $q_i$ is on the bottom side of $R$ for all $i\in \mathbb{N}^+$.
By Lemma~\ref{lem:shallow}, there exists a shallow-light tree $T_i$ with parameter $\eps'=\eps/4$ of weight $O(\hght(R)/2^i)$ from the root $p_i$ to the line segment $oq_i$. Between any point $a\in p_{i-1}p_i$ and $b\in q_iq_{i+1}$, we can combine a
vertical segment $ap_i$ with a path in the tree $T_i$ from $p_i$
to $b$ to obtain a path of weight at most $(1+\eps')\|ab\|$.

The weight of the union of the trees $T_i$, for all $i\in \mathbb{N}^+$, is
$O(\sum_{i=1}^{\infty}\hght(R)/2^i)=O(\hght(R))$. In fact, we do not need
infinitely many trees, since $S$ is finite, hence it contains a finite number
of point in the interior of the left side of $R$. It suffices to construct the
trees $T_i$, $i=1,\ldots , m$, such that all points in $S$ in the left side of
$R$ are at or above $p_m$.
\end{proof}

\begin{figure}[htbp]
\centering
\includegraphics[width=\textwidth]{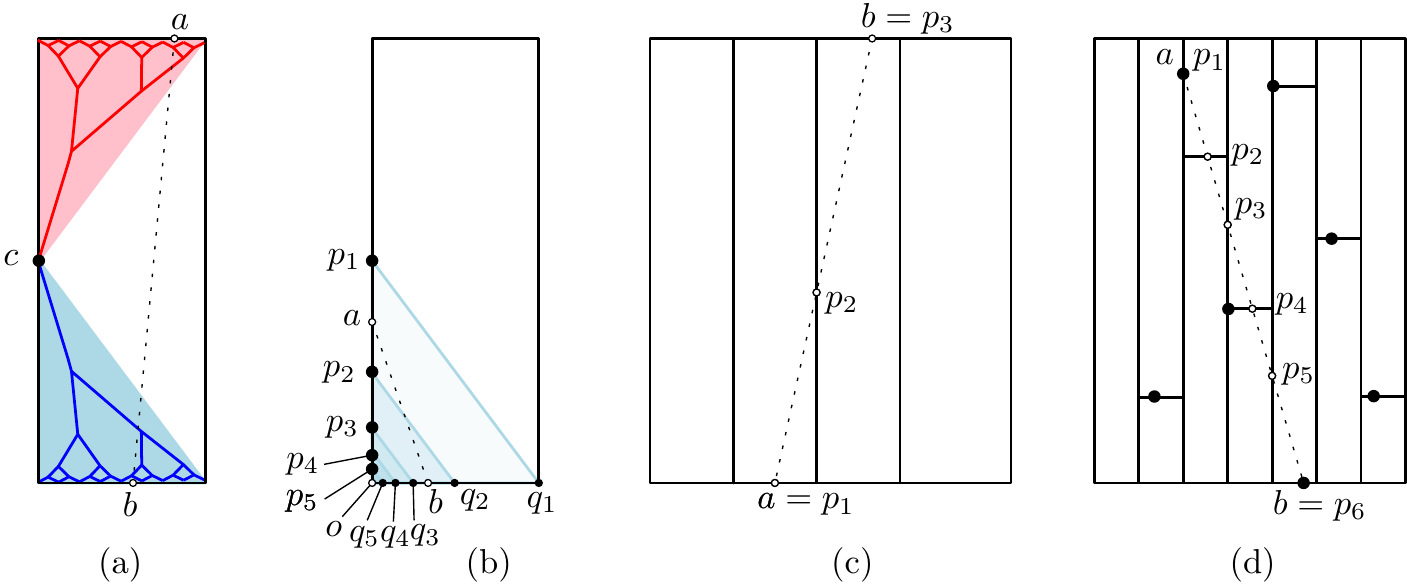}
\caption{(a) A rectangle $R$ with $\frac18\,\sqrt{\eps}\leq \aratio(R)\leq \frac14\,\sqrt{\eps}$,
and two shallow-light trees rooted at the midpoint $c$ of the left side of $R$.
(b) A sequence of shallow-light trees rooted at $p_1,\ldots ,p_m$ on the left side of $R$.
(c) A subdivision of a rectangle $R$.  (d) A rectangulation of the bounding box of $S$ into rectangels $R$ with $ \aratio(R)\leq \frac14\,\sqrt{\eps} <\aratio(R)$.}
\label{fig:rectangle}
\end{figure}

We can generalize Lemma~\ref{lem:rect1} to axis-aligned rectangles of arbitrary aspect ratio.

\begin{lemma}
\label{lem:rect2}
Let $R$ be an axis-aligned rectangle. Then for a finite point set $S$ on the boundary of $R$,
and for the interval $D=[\frac{\pi-\sqrt{\eps}}{8},\frac{\pi+\sqrt{\eps}}{8}]$ of directions, there exists a directional $(1+\eps)$-spanner $G$ of weight $O(\hght(R)+\wdth(R)/\sqrt{\eps})$.
\end{lemma}
\begin{proof}
If $\frac18\, \sqrt{\eps}\leq \aratio(R)\leq \frac14\, \sqrt{\eps}$, then Lemma~\ref{lem:rect1}
completes the proof. Otherwise, we greedily subdivide $R$ by parallel lines as follows.
First assume that $\frac14\, \sqrt{\eps}< \aratio(R)$; refer to Fig.~\ref{fig:rectangle}(c).
For a vertical line $L$, denote by $L^-$ and $L^+$, resp.,
the halfplanes on the left and right of $L$.
Let $L$ be the leftmost vertical line such that $\aratio(R\cap L^-)=\frac18\,\sqrt{\eps}$.
Then subdivide $R$ into $R\cap L^-$ and $R\cap L^+$ by a vertical segment
of weight $\hght(R)$, and recurse on $R\cap L^+$.
All subdivision edges are vertical, of weight $\hght(R)$, and their total weight is
\begin{equation}\label{eq:cut1}
\hght(R)\cdot \left\lfloor \frac{\wdth(R)/\hght(R)}{\frac18\,\sqrt{\eps}}\right\rfloor
\leq O\left(\frac{\wdth(R)}{\sqrt{\eps}}\right).
\end{equation}

Similarly, if $\aratio(R)<\frac18\, \sqrt{\eps}$, we can greedily subdivide $R$ by horizontal lines into axis-parallel rectangles.
For a horizontal line $L$, denote by $L^\uparrow$ and $L^\downarrow$, resp.,
the halfplanes above and below $L$. Let $L$ be the highest horizontal line such that $\aratio(R\cap L^\uparrow)=\frac14\,\sqrt{\eps}$.
Then subdivide $R$ into $R\cap L^\uparrow$ and $R\cap L^\downarrow$ by a horizontal segment
of weight $\wdth(R)$, and recurse on $R\cap L^\downarrow$.
In this case, all subdivision edges are horizontal, they all are of weight $\wdth(R)$, and their total weight is
\begin{equation}\label{eq:cut2}
\wdth(R)\cdot \left\lfloor \frac{\frac14\,\sqrt{\eps}}{\wdth(R)/\hght(R)}\right\rfloor
\leq O\left(\sqrt{\eps}\cdot\hght(R)\right).
\end{equation}

Assume that $R$ has been subdivided into rectangles $R_1,\ldots ,R_k$, such that
$\frac18\, \sqrt{\eps}\leq \aratio(R_i)\leq \frac14\, \sqrt{\eps}$.
For every $i\in \{1,\ldots, k\}$, let $S_i$ be the set of intersection points between the boundary of $R_i$ and the line segments spanned by $S$. For the point set $S_i$ and the directions $D$, Lemma~\ref{lem:rect1} yields a directional $(1+\eps)$-spanner $G_i$ of weight  $\|G_i\|=\hght(R_i)$.

Let $G=\bigcup_{i=1}^k G_i$. From \eqref{eq:cut1} and \eqref{eq:cut2}, we get
$\|G\|=\sum_{i=1}^k \|G_i\| =O(\sum_{i=1}^k\hght(R_i)) =
O(\hght(R)+\wdth(R)/\sqrt{\eps})$, as required.
It remains to show that $G$ is a directional $(1+\eps)$-spanner.
Let $a,b\in S$ with direction in $D$; see Fig.~\ref{fig:rectangle}(c). The vertical edges of $R_1,\ldots , R_k$ subdivide $ab$ into a path of collinear line segments $a=p_0,\ldots ,p_\ell=b$. Each segment $p_{i-1}p_i$ lies in some rectangle $R_j$ between points $p_{i-1},p_i\in S_j$, and so $G_j$ contains a $p_{i-1}p_i$-path of weight at most $(1+\eps)\|p_{i-1}p_i\|$.
The concatenation of these paths is an $ab$-path of weight at most $\sum_{i=1}^\ell (1+\eps)\|p_{i-1}p_i\|= (1+\eps)\|ab\|$.
\end{proof}

\subsection{Rectangulations}
\label{ssec:rectangulations}

Let $\per(P)$ denote the perimeter of $P$. A polygon $P$ is \emph{rectilinear} if every edge is horizontal or vertical. A \emph{rectangulation} of polygon $P$ is a subdivision of $P$ into axis-parallel rectangles. De~Berg and van~Kreveld~\cite{BergK94} proved that for a rectilinear simple polygon $P$ with $n$ vertices, one can efficiently compute a rectangulation of weight $O(\per(P)\log n)$, and this bound is the best possible (already for stair-case polygons).

For an arbitrary set $S$ of $n$ points in the plane, we can construct a rectangulation of
the axis-aligned bounding box of $S$ with weight $O(\|\MST\|\log n)$.
Combining such a rectangulation with Lemma~\ref{lem:rect2}, we are now ready to prove Lemma~\ref{lem:dir}.

\directionallemma*
\begin{proof}
Let $T$ be the rectilinear MST of $S$, that is, a spanning tree of minimum weight w.r.t.\ $L_1$-norm, realized in the plane such that every edge is an L-shape (the union of a horizontal and a vertical segment). It is well known that $\|T\|\leq \sqrt{2}\, \|\MST(S)\|$. Let $R$ be the minimum axis-aligned bounding box of $T$.
By the minimality of $R$, the boundary of $R$ contains at least two vertices of $T$.
Consequently, $T\cup \partial R$ is connected, and it subdivides the interior of $R$
into rectilinear simple polygons (\emph{faces})
of total weight at most $2(\|T\|+\per(R))=O(\|\MST(S)\|)$.

By the result of de~Berg and van~Kreveld~\cite{BergK94}, we can rectangulate each face of
$T\cup \partial R$. Let $\mathcal{R}$ denote the resulting rectangulation.
Since every face has $O(n)$ vertices, and every edge
is on the boundary of at most two faces, the total perimeter of the rectangles in $\mathcal{R}$
is $\sum_{R\in \mathcal{R}}\per(R)=O((\|T\|+\per(R))\log n)=O(\|\MST(S)\|\log n)$.

For every $R\in \mathcal{R}$, let $S(R)$ be the set of intersection points between the boundary of $R$ and the line segment induced by $S$. By Lemma~\ref{lem:rect2}, there exists a directional Euclidean Steiner $(1+\eps)$-spanner $G(R)$ for $S(R)$ of weight $O(\per(R)\cdot \eps^{-1/2})$. Let $G=\bigcup_{R\in \mathcal{R}} G(R)$. Its total weight $\|G\|=\sum_{R\in \mathcal{R}} O(\per(R)\cdot \eps^{-1/2}) = O(\|\MST(S)\|\, \eps^{-1/2}\log n)$.
We can verify that $G$ is a directional $(1+\eps)$-spanner for $S$, similarly to the proof of Lemma~\ref{lem:rect2}. Let $a,b\in S$ such that the directions of $ab$ is in $D$; see Fig.~\ref{fig:rectangle}(d).
The rectangulation subdivides $ab$ into a path of collinear segments $a=p_0,\ldots ,p_\ell=b$. Each segment $p_{i-1}p_i$ lies in some rectangle $R\in \mathcal{R}$ between points $p_{i-1},p_i\in S(R)$, and so $G(R)$ contains a $p_{i-1}p_i$-path of weight at most $(1+\eps)\|p_{i-1}p_i\|$. The concatenation of these paths is an $ab$-path of weight at most $\sum_{i=1}^\ell (1+\eps)\|p_{i-1}p_i\|= (1+\eps)\|ab\|$, as required.
\end{proof}

\begin{remark}
The $\log(n)$-factor in Theorem~\ref{thm:upper} is due to the rectangulations of rectilinear  polygons with $O(n)$ vertices.
Instead of rectangulations, one could use a minimum-weight Steiner subdivision into convex faces (assuming that Lemmas~\ref{lem:rect1}--\ref{lem:rect2} generalize to \emph{convex} polygons). However, this approach would not yield more than a $\log \log (n)$-factor improvement. Dumitrescu and T\'oth~\cite{DumitrescuT11} probed that every simple polygon $P$ with $n$ vertices admits a convex subdivision of weight $O(\per(P)\log n/\log \log n)$, and this bound is the best possible.
\end{remark}

\begin{remark}
Instead of a rectangulation, one could also use a subdivision into histograms.
A \emph{histogram} is a rectilinear simple polygon bounded by three axis-parallel line segments and one $x$- or $y$-monotone path. A classical \emph{window partition}~\cite{Link00,Suri90} subdivides a simple rectilinear polygon $P$ into histograms such that every axis-parallel line segment in $P$ intersects (\emph{stabs}) at most three histograms~\cite{Edelsbrunner1984167,Levcopoulos}. Due to the stabbing property, the total perimeter of the resulting histograms is $O(\per(P))$. For a point set $S$, this approach yields a histogram subdivision of the bounding box of $S$ with  weight $O(\|\MST(S)\|)$.
Very recently, Bhore and T\'{o}th~\cite{DBLP:journals/corr/abs-2012-02216} improved the upper bound $O(\eps^{-1}\log n)$ of Theorem~\ref{thm:upper} to $O(\eps^{-1})$ by combining directional spanners with a modified window partition.
\end{remark}

\enlargethispage{-1\baselineskip}
\section{Conclusions}
\label{sec:cons}

In this paper, we have studied Euclidean Steiner $(1+\epsilon)$-spanners under two optimization criteria, \emph{lightness} and \emph{sparsity}, and provided improved lower and upper bounds. Our upper bound of $O(\eps^{-1}\log n)$ on the minimum lightness of Steiner $(1+\eps)$-spanners in the Euclidean plane has recently been improved to the bound $O(\eps^{-1})$ in~\cite{DBLP:journals/corr/abs-2012-02216}, matching the lower bound of $\Omega(\eps^{-1})$ of Theorem~\ref{thm:lb}.
However, for lightness in dimensions $d\geq 3$, an $\tilde{\Theta}(\eps^{1/2})$-factor gap remains between the current upper bound $\tilde{O}(\eps^{-(d+1)/2})$~\cite[Theorem~1.6]{le2020unified} and the lower bound $\Omega(\eps^{-d/2})$ of Theorem~\ref{thm:lb}.

In Euclidean $d$-space, the same point sets (grids in two parallel hyperplanes) establish the lower bounds $\Omega(\eps^{-d/2})$ for lightness and $\Omega(\eps^{(1-d)/2})$ for sparsity for Steiner $(1+\eps)$-spanners (cf.~Theorem~\ref{thm:lb}).
Le and Solomon constructed spanners with sparsity $\tilde{O}(\eps^{(1-d)/2})$~\cite[Theorem~1.3]{le2019truly}, matching the lower bound in every dimension $d\in \mathbb{N}$, but the lightness of these spanners is significantly higher.
In dimensions $d\geq 3$, they construct spanners with lightness $\tilde{O}(\eps^{-(d+1)/2})$~\cite[Theorem~1.6]{le2020unified},
but these spanners have significantly higher sparsity.

We conjecture that a Euclidean Steiner $(1+\eps)$-spanner cannot simultaneously attain both lower bounds (lightness and sparsity) of Theorem~\ref{thm:lb}. Therefore, exploring trade-offs between lightness and sparsity in Euclidean $d$-space remains an open problem.

A critical aspect of graphs is their embeddibility in low-genus surfaces. A geometric graph in $\mathbb{R}^2$ is a \emph{plane graph} if no two edges cross each other. Note that every Steiner spanner $G=(V,E)$ for a point set $S$ can be turned into a plane graph (with the same stretch factor ratio and the same weight) by introducing Steiner points at every edge crossing. However, the number of new Steiner points would be proportional to $O(|E|^2)$, which is prohibitive. It remains an open problem to bound the sparsity of a \emph{plane} Steiner $(1+\eps)$-spanner for $n$ points in Euclidean plane, as a function of $n$ and $\eps$.
\pagebreak

Angles and directions play a crucial role in our lower bound analysis (Section~\ref{sec:lower}) and upper bound construction (Section~\ref{sec:upper}). While angles are invariant under rotations only in Euclidean spaces, they can be defined in any inner product space, such as $\mathbb{R}^d$ under $L_p$ norm, for $p\geq 2$. We leave it as an open problem to derive upper and lower bounds on the lightness and sparsity of Steiner $(1+\eps)$-spanners in other inner product spaces.
\enlargethispage{1.5\baselineskip}

\bibliographystyle{plainurl}
\bibliography{p013-Bhore}

\end{document}